\documentclass[runningheads]{llncs}

\usepackage [T1]   {fontenc}
\usepackage [utf8] {inputenc}

\usepackage[
  scaled = .96
]{PTSansNarrow}
\newcommand*{\narrowfont}{\fontfamily{PTSansNarrow-TLF}\selectfont}

\usepackage[
 svgnames,
 table,
]{xcolor}

\usepackage[
 colorlinks,
 urlcolor = NavyBlue,
 citecolor = NavyBlue,
 linkcolor = NavyBlue,
]{hyperref}

\usepackage{
  amsmath,
  amsfonts,
  amssymb,
  cite,
  ebproof,
  float,
  graphicx,
  ifthen,
  kantlipsum,
  listings,
  mathpartir,
  microtype,
  multirow,
  nicefrac,
  stmaryrd,
  textcomp,
  utfsym,
  wrapfig,
  xcolor,
  xparse,
  xspace
}

\usepackage[
  scale = .8
]{noto-mono}

\usepackage[
  inline
]{enumitem}

\usepackage{cleveref}

\crefname{equation}{Eq.}{equations}
\Crefname{equation}{Equation}{Equations}
\crefname{proposition}{Prop.}{propositions}
\Crefname{proposition}{Proposition}{Propositions}
\crefname{lemma}{Lemma}{lemmata}
\Crefname{lemma}{Lemma}{Lemmata}
\crefname{listing}{Listing}{listings}
\Crefname{listing}{Listing}{Listings}
\crefname{definition}{Def.}{definitions}
\Crefname{definition}{Definition}{Definitions}
\crefname{theorem}{Thm.}{theorems}
\Crefname{theorem}{Theorem}{Theorems}
\crefname{figure}{Fig.}{figures}
\Crefname{figure}{Figure}{Figures}
\crefname{page}{p.}{pages}
\Crefname{page}{Page}{Pages}
\crefname{section}{Sect.}{sections}
\Crefname{section}{Section}{Sections}

\usepackage[normalem]{ulem}

\usepackage[draft,silent]{fixme}
\fxsetup{theme=color,mode=multiuser}
\definecolor{fxtarget}{rgb}{0.8000,0.0000,0.0000}
\FXRegisterAuthor{todo}{TODO}{TODO}
\FXRegisterAuthor{ebj}{EBJ}{Einar}
\FXRegisterAuthor{ek}{EK}{Eduard}
\FXRegisterAuthor{aw}{AW}{Andrzej}
\FXRegisterAuthor{rp}{RP}{Ra\'ul}

\colorlet{keywordcolor}{blue!50!black}
\colorlet{commentcolor}{green!60!black}
\colorlet{typecolor}{violet}
\newcommand{\sourcefont}{\ttfamily\small}
\newcommand{\commentfont}{\slshape\rmfamily\color{commentcolor}}

\lstdefinelanguage{ABS}{
        keywords={od,do,assert,old,last,this,new,data,type,def,case,of,local,class,interface,
        extends,implements,if,then,else,await,get,return,skip,while,module,
        import,export,from,to,suspend,delta,adds,modifies,removes,original,productline,
        features,core,corefeatures,optionalfeatures,after,when,product,hasAttribute,
        hasMethod,hasField,hasInterface,uses,root,extension,group,allof,oneof,require,
        stateupdate,object,main,objectupdate,classupdate,fi,
        exclude,original,ifin,ifout,opt,null,
        newgroup,data,thiscomp,in,joins,leaves,subtypeOf,acquire,except,as,component,Pre,Abs
        },
        keywordstyle=\color{keywordcolor}\bfseries\sffamily,
        morekeywords=[2]{Unit, Fut,Int, Bool, Rat, List, Set, Pair, Fut, Maybe, String, Triple, Either, Map},
        keywordstyle=[2]\color{typecolor},
        sensitive=true,
        comment=[l]{//},
        morecomment=[s]{/*}{*/},
        morestring=[b]"
}

\lstdefinestyle{codeabs}{
        basicstyle=\sourcefont\upshape,
        keywordstyle=\color{keywordcolor}\bfseries\sffamily,
        commentstyle=\commentfont,
        columns=fullflexible,
        mathescape=true,
        escapechar={\#},
        keepspaces=true,
        showstringspaces=false,
        aboveskip=8pt, 
        numbers=left,
        stepnumber=1, 
        numberstyle=\ttfamily\scriptsize\color{gray},
        numbersep=4pt,
        xleftmargin=1.5em,
        xrightmargin=1.5em,
        framexleftmargin=1.2em,
        framexrightmargin=1em,
        framextopmargin=0.5ex,
        breaklines=true,
        breakindent=3pt,
}

\lstdefinestyle{abs}{
        style=codeabs,
        language=ABS,
}

\newcommand{\abs}[2][]{\lstinline[style=abs,#1]{#2}}

\lstnewenvironment{abscode}[1][]{
        \minipage{1\linewidth}
        \lstset{style=abs,
        backgroundcolor=\color{white},
        rulecolor=\color{gray!50},
        frame=tblr,
        captionpos=b,
        #1}
}{\endminipage}

\renewcommand \UrlFont \sffamily

\def\orcidID#1{\smash{\href{http://orcid.org/#1}{\protect\raisebox{-1.25pt}{\protect\includegraphics{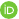}}}}}

\NewDocumentCommand{\Expectation}%
  {O{\epsilon} D(){absent} m}%
  {\ifthenelse{ \equal {#2} {} }
    {\ensuremath{\mathbf{E}_{#1}{\left({#3}\right)}}}
    {\ensuremath{\mathbf{E}_{#1}{{#3}}}}}

\NewDocumentCommand{\ExpectedV}%
  {O{\epsilon,\pi} D(){absent} m}%
  {\ifthenelse{ \equal {#2} {} }
    {\ensuremath{\mathbb{E}_{#1}{\left({#3}\right)}}}
    {\ensuremath{\mathbb{E}_{#1}{{#3}}}}}

\newcommand{\etal}{\emph{et al.}\xspace}

\renewcommand{\phi}{\varphi}

\newcommand{\embed}[1]{\ensuremath{{[\![#1]\!]}}}
\renewcommand{\epsilon}{\varepsilon}
\newcommand{\overbar}[1]{\mkern 1.5mu\overline{\mkern-1.5mu#1\mkern-1.5mu}\mkern 1.5mu}
\newcommand{\dlbox}[1]{[#1]}

\newcommand{\final}[1]{\textrm{final}(#1)}

\newcommand{\MDP}{\mathcal{M}}
\newcommand{\dom}[1]{\ensuremath{\textnormal{dom}\,#1}}
\newcommand{\PDL}{\textnormal{\narrowfont pDL}\xspace}
\newcommand{\pGCL}{\textnormal{\narrowfont pGCL}\xspace}
\newcommand{\FOL}{\textnormal{\narrowfont FOL}\xspace}
\newcommand{\LIT}{\ensuremath{\textnormal{\narrowfont ATF}}}
\newcommand{\pchoice}[1]{{\,}_{#1}\!\oplus }
\newcommand{\ndchoice}{\ensuremath{\sqcap}}
\newcommand{\pbox}[2]{\ensuremath{[#1]_{#2}}}
\newcommand{\pdleq}[3]{\langle #3 \rangle}

\newcommand{\subst}[2]{\ensuremath{[ #1 := #2 ]}}

\newcommand{\updates}{\mathcal{U}}

\newcommand{\func}[1]{{\boldsymbol{#1}}}
\newcommand{\pval}{\ensuremath{\mathit{p}}}
\newcommand{\pfun}{\ensuremath{\func{p}}}
\newcommand{\pfunepsilon}{\ensuremath{\pfun(\epsilon)}}

\newcommand{\gmid}{~\mid~}
\newcommand{\updatesemantics}[1]{\llbracket{#1}\rrbracket}
\newcommand{\judgment}[2]{{#1}\vdash{#2}}

\newcommand{\ncondrule}[3]{
  \begin{array}{c}
    \textsc{ ({#1})} \\[1pt]
    #2 \\[1pt]
    \hline\\[-7pt]
    #3
  \end{array} }

\newcommand{\State}{\textit{State}\xspace}
\newcommand{\Act}{\textit{Act}}

\newcommand{\code}[1]{\text{\lstinline[mathescape=true]|#1|}}

\newcommand{\rulename}[1]{\textsc{\scriptsize(#1)}}
\newcommand{\refrule}[1]{\textsc{\footnotesize #1}}
\newcommand{\INFER}[4][]{\frac{\begin{array}{c}\rulename{#2}\\#3
\\[0.5mm]
\end{array}
}{
\begin{array}{c}
\\[-3.5mm]
\displaystyle{#4}
\end{array}
}\;{#1}}

\newcommand{\rulemacro}[3][]{\frac{\begin{array}{c}\\#2
\\[0.5mm]
\end{array}
}{
\begin{array}{c}
\\[-3.5mm]
\displaystyle{#3}
\end{array}
}\;{#1}}

\title{Towards a Proof System for Probabilistic~Dynamic~Logic}
\author{Einar Broch Johnsen \inst{1} \orcidID{0000-0001-5382-3949} \and
Eduard Kamburjan \inst{1} \orcidID{0000-0002-0996-2543} \and
Raul Pardo \inst{2} \orcidID{0000-0003-0003-7295} \and\\
Erik Voogd \inst{1} \orcidID{0009-0007-9712-3224} \and
Andrzej Wąsowski \inst{2} \orcidID{0000-0003-0532-2685}
}
\authorrunning{E.\,B. Johnsen et al.}
\institute{Department of Informatics, University of Oslo, Oslo, Norway
  \\ \email{\{einarj,eduard,erikvoogd\}@ifi.uio.no}
  \and
IT University of Copenhagen, Copenhagen, Denmark\\
\email{\{raup,wasowski\}@itu.dk}}

\begin{document}

\maketitle              

\begin{abstract}
  Whereas the semantics of probabilistic languages has been
  extensively studied, specification languages for their properties
  have received less attention---with the notable exception of recent
  and on-going efforts by Joost-Pieter Katoen and collaborators. In
  this paper, we revisit probabilistic dynamic logic (\PDL), a
  specification logic for programs in the probabilistic guarded
  command language (\pGCL) of McIver and Morgan. Building on dynamic
  logic, \PDL can express both first-order state properties and
  probabilistic reachability properties. In this paper, we report on
  work in progress towards a deductive proof system for \PDL. This
  proof system, in line with verification systems for dynamic logic
  such as KeY, is based on forward reasoning by means of symbolic
  execution.
\end{abstract}


\section{Introduction}\label{sec:intro}
Joost-Pieter Katoen has pioneered techniques for the verification of
probabilistic systems, including numerous contributions on
model-checking algorithms (e.g.,
\cite{baier99concur,baier03tse,junges24fmsd}), including tools such as
the probabilistic model checker Storm
\cite{dehnert017cav,hensel22sttt}, as well as proof systems for
deductive verification (e.g.,
\cite{batz21popl,batz23tacas,schroer23oopsla,feng23oopsla,batz24popl,mcIver18popl,kaminski18jacm}). Whereas
this line of work is rooted in Hoare logics and weakest precondition
calculi, we here consider a specification language for probabilistic
systems based on dynamic logic \cite{harel00dynlog}: probabilistic
dynamic logic (\PDL for short) \cite{2022-pdl-ictac}.  We believe it
is interesting to study specification languages and deductive
verification based on dynamic logic because dynamic logic is strictly
more expressive than Hoare logic and weakest precondition calculi; in
fact, both can be embedded in dynamic logic \cite{Haehnle22a}. In
contrast to these calculi, dynamic logics are closed under logical
operators such as first-order connectives and quantifiers; for
example, program equivalence, relative to state formulae $\phi$ and
$\psi$, can be expressed by the dynamic logic formula
$\phi \rightarrow \pbox{s_1}{} \psi \iff \phi \rightarrow \pbox{s_2}{}
\psi$. Consequently, specification languages based on dynamic logic,
including \PDL, have classical model-theoretic semantics known from
logics: a satisfaction semantics.

This paper revisits \PDL and its model-theoretic semantics, and its
main contribution is a deductive verification system for \PDL based on
forward reasoning about \PDL judgments, in contrast to the backwards
reasoning used in weakest-precondition-based approaches. To this aim,
we sketch a proof system for \PDL based on symbolic execution rules
that collect constraints about probabilities, and a prototype
implementation of this proof system using
Crowbar~\cite{DBLP:journals/scp/KamburjanSR23}, a modular symbolic
execution engine, and the SMT solver Z3 \cite{demoura2008} to solve
probabilistic constraints.

\section{Motivating Example}\label{sec:example}

As an example of a probabilistic program in \pGCL, consider the
\emph{Monty Hall game}, in which a host presents three doors to a
player. The first door contains a prize and the other doors are
empty. The player needs to decide (or guess) the door behind which the
prize is hidden.  The game proceeds as follows.  First, the location
of the prize is non-deterministically selected by the host.  Then, the
player selects a door.  The host opens an empty door that was not
selected by the player, who is asked whether she would like to switch
doors.  We determine, using \PDL, what option increases the chances of
winning the prize (switching or not).

\begin{figure}[t]
\begin{lstlisting}[mathescape=true,xleftmargin=.25\textwidth]
prize := 0 $\ndchoice$ (prize := 1 $\ndchoice$ prize := 2);
choice := 0 $\pchoice{1/3\;}$ (choice:=1 $\pchoice{1/2\;}$ choice:=2);
if (prize = choice)
    open := (prize+1)%3 $\ndchoice$ open := (prize+2)%3;
else
    open := (2*prize-choice)%3;
if (switch)
    choice := (2*choice-open)%3
else
    skip
  \end{lstlisting}
  \caption{The Monty Hall Program in \pGCL (\lstinline|Monty_Hall|).}\label{code:monty-hall}
  \label{code:monty}
\end{figure}

\Cref{code:monty-hall} shows a \pGCL\ program, \lstinline|Monty_Hall|,
modeling the behavior of host and player.
The program contains four variables: \lstinline|prize| (the door
hiding the prize), \lstinline|choice| (the door selected by the
player), \lstinline|open| (the door opened by the host),
\lstinline|switch| (a Boolean indicating whether the user switches
door in the last step).
Note that the variable \lstinline|switch| is undefined in the program,
and will encode the strategy of the player.
Line 1 models the host's non-deterministic choice of the door for the prize.
Line 2 models the player's choice of door (uniformly over the three doors).
Lines 3--6 model the selection of the door to open, from the doors
that were not selected by the player.
Lines 7--10 model whether the player switches door or not.
For simplicity, we use a slight shortcut to compute the door to open and to
switch in Lines 6 and 8, respectively.
Note that for $x,y \in \{0,1,2\}$, the expression
$z = (2x-y)\;\mathbf{mod}\;3$ simply returns $z \in \{0,1,2\}$ such
that $z \not= x$ and $z\not=y$.
Similarly, the expressions
$y=(x+1)\;\mathbf{mod}\;3, z=(x+2)\;\mathbf{mod}\;3$ in Line~4 ensure
that $y\not=x$, $z\not=x$ and $y \not= z$.
This shortcut computes the doors that the host may open when the
player chose the door with the prize in Line~2.

An example of \PDL specification of the Monty Hall game is the
following formula:
\begin{equation}
  \textit{switch} = \textit{true} \rightarrow [\text{\lstinline|Monty_Hall|}]_\pfun(\textit{choice} = \textit{prize}).
  \label{eq:monty}
\end{equation}
This formula expresses that if the player's strategy is to change door
(i.e., the program's precondition is given by the state formula
$\textit{switch} = \textit{true}$), then the probability of reaching a
state characterized by the postcondition
$\textit{choice} = \textit{prize}$ after successfully executing
\lstinline|Monty_Hall|, is $\pfun$. But what should be the value of
$\pfun$? We will show in \Cref{sec:implementation} that we can
prove this specification for $\pfun = \min(\pfun_0,\pfun_1,\pfun_2)$
where each $\pfun_i$ is the probability for the different locations of
the prize.

\section{Preliminaries}
We briefly introduce the programming language
\pGCL~\cite{mciver05book}, but first we recall Markov Decision
Processes~\cite{baier08book,puterman}, which we use to define the
semantics of \pGCL.

\subsection{Markov Decision Processes}
Markov decision processes are computational structures that feature
both probabilistic and nondeterministic choice. We start from the
following definition (e.g., \cite{baier08book,puterman}):

\begin{definition}[Markov Decision Process]\label{def:mdp}
  A \emph{Markov Decision Process} (MDP) is a tuple $M=(\State,\Act, \mathbf{P})$ where
  \begin{enumerate*}[label=(\roman*)]

    \item $\State$ is a countable set of states,

    \item $\Act$ is a countable set of actions, and

    \item $\mathbf{P}: \State \times \Act \rightarrow \textrm{Dist}(\State)$ is a partial transition probability function.

  \end{enumerate*}
\end{definition}

Let \(\sigma\) denote the states and $a$ the actions of an MDP. A
state \(\sigma\) is \emph{final} if no further transitions are
possible from it, i.e.\
\((\sigma, a) \not \in \textrm{dom} (\mathbf{P})\) for any $a$. A
\emph{path}, denoted \(\overline\sigma\), is a sequence of states
$\sigma_1,\ldots,\sigma_n$ such that $\sigma_n$ is final and there are
actions \(a_1, \ldots, a_{n-1}\) such that
\(\mathbf{P}(\sigma_i,a_i)(\sigma_{i+1})\geq 0\) for $1\leq i <
n$. Let \(\final{\overbar\sigma}\) denote the final state of a path
\(\overbar\sigma\).  To resolve non-deterministic choice, a
\emph{positional policy} \(\pi\) maps states to actions, so
\(\pi: \State \rightarrow \Act\).  Given a policy \(\pi\), we define a
transition relation
$\xrightarrow{\cdot}_\pi \subseteq \State \times [0,1] \times \State$
on states that resolves all the demonic choices in $\mathbf{P}$ and
write
$$
\sigma\xrightarrow{\pval_i}_\pi \sigma' \,\, \text{iff} \,\,
  \mathbf{P}(\sigma,\pi(\sigma)) (\sigma') = \pval_i.
$$
Similarly, the reflexive and transitive closure of the transition
relation,
$\smash{\xrightarrow{\pval}}^\ast_\pi \subseteq \State \times [0,1]
\times \State$, defines the probability of a path as
\begin{equation}
  p = \Pr (\overline\sigma) = 1 \cdot \pval_1 \cdots \pval_n
  \quad \text{ where }  \sigma_1 \xrightarrow{\pval_1}_\pi \cdots \xrightarrow{\pval_n}_\pi \sigma_n .
\end{equation}

Thus, a path with no transitions consists of a single state $\sigma$,
and $\Pr(\sigma) = 1$. Let $\textrm{paths}_\pi(\sigma)$ denote the set
of all paths with policy $\pi$ from $\sigma$ to final states.
\looseness -1

An MDP may have an associated \emph{reward function}
\(r: \State \to [0,1]\) that assigns a real value \(r(\sigma)\) to any
state \( \sigma \in \State \).  (In this paper we assume that rewards
are zero everywhere but in the final states.) We define the
\emph{expectation} of the reward starting in a state \(\sigma\) as the
greatest lower bound on the expected value of the reward over all
policies; so the real valued function defined as \looseness = -1
\begin{equation}
  \Expectation[\sigma]()r = \inf_\pi \ExpectedV[\sigma,\pi]()r = \inf_\pi \sum_{\overline \sigma \in \textrm{paths}_\pi(\sigma)} \!\!\!\!\!\! \Pr (\overline \sigma) \, r ( \final{\overline \sigma} )  \enspace , \label{eq:expectation}
\end{equation}
where \ExpectedV[\sigma, \pi]()r stands for the \emph{expected value}
of the random variable induced by the reward function under the given
policy, known as the \emph{expected reward}.

In this paper, we assume that MDPs (and the programs we derive them
from) arrive at final states with probability 1 under all policies.
This means that the logic \PDL that we will be defining and
interpreting over these MDPs can only talk about properties of almost
surely terminating programs, so in general it cannot be used to reason
about termination without adaptation.

\subsection{pGCL: A Probabilistic Guarded Command Language}
\label{subsec:pgcl}

As programming language, we consider the probabilistic guarded command
language (\pGCL) of McIver and Morgan~\cite{mciver05book}, a core
language which features both probabilistic and non-deterministic
choice. We briefly recall the syntax of \pGCL and it semantics,
formulated as a probabilistic transition system.

\subsubsection{Syntax of  \pGCL.}
Let $X$ be a set of program variables and $x \in X$, the syntax of
\pGCL is defined as follows:
\begin{equation*}
  \begin{array}{lrl}
    v &::= & \textit{true} \mid \textit{false} \mid 0 \mid 1 \mid \ldots \\[.6mm]
    e &::=& v \mid x \mid \textit{op}\ e \mid e\ \textit{op}\ e\\[.6mm]
    op& ::= & +\; \mid \; -\; \mid \; *\; \mid \; /\; \mid \; >\; \mid \; ==\; \mid \; \geq\\[.6mm]
    s &::= & s \ndchoice s \mid s \pchoice{e} s \mid s;s \mid \code{skip} \mid x:= e
             \mid \code{if}\ e\ \{ s \}\ \code{else}\ \{ s \} \mid \code{while}\ e\ \{ s \}
  \end{array}
\end{equation*}
Statements $s$ include the \emph{non-deterministic (or demonic)
  choice} $s_1 \sqcap s_2$ between branches $s_1$ and $s_2$, and
$s \pchoice{e} s'$ for the \emph{probabilistic choice} between $s$ and
$s'$. A non-deterministic program $s_1 \sqcap s_2$ will arbitrarily
select a branch for execution, whereas in a probabilistic program
$s \pchoice{e} s'$, if the expression $e$ evaluates to a value
$\pval \in [0,1]$ given the current values for the program variables,
then $s$ and $s'$ have probability $\pval$ and $1-\pval$ of being
selected, respectively.  Binary operators and the remaining statements
have the usual meaning.

\subsubsection{Semantics of \pGCL.}

The semantics of a \pGCL\ program $s$ can now be defined as an MDP
$\MDP_s$. A state $\sigma$ of $\MDP_s$ is a pair of a \emph{valuation}
and a \emph{program}, so $\sigma=\langle \epsilon, s\rangle$ where the
valuation $\epsilon$ is a mapping from all the program variables in
$s$ to concrete values (we may omit the program from this pair, if it
is unambiguous in the context).  The state
$\langle \epsilon, s\rangle$ represents an \emph{initial state} of the
program $s$ given some initial valuation $\epsilon$ and the state
$\langle\epsilon, \code{skip}\rangle$ represents a \emph{final state}
in which the program has terminated with the valuation $\epsilon$. For
a concrete program, the \emph{policy} $\pi$ is a function that
determines how non-deterministic choice is resolved for a given
valuation of the program variables; i.e.,
$\pi\langle\varepsilon, s_1 \sqcap s_2\rangle= s_i$ for either $i=1$ or
$i=2$.  The rules defining the partial transition probability function
for a given policy $\pi$ are shown in \Cref{fig:semantics}.

\begin{figure}[t]
\centering
$$\begin{array}{c}
\ncondrule{Assign}{\epsilon'=\epsilon[x\mapsto \epsilon(e)]}{\langle \epsilon, x:=e\rangle \xrightarrow{1}_\pi  \langle \epsilon', \code{skip}\rangle}
\qquad
\ncondrule{Composition1}{
        \langle \epsilon, s_2 \rangle\xrightarrow{p}_\pi  \langle \epsilon',s\rangle}{\langle \epsilon, \code{skip}; s_2 \rangle
      \xrightarrow{p}_\pi  \langle \epsilon',s\rangle}
\qquad
\ncondrule{Composition2}{\langle \epsilon, s_1 \rangle \xrightarrow{p}_\pi  \langle \epsilon',s\rangle}{\langle \epsilon, s_1; s_2 \rangle
   \xrightarrow{p}_\pi  \langle \epsilon',s;s_2\rangle}
\\\\
\ncondrule{DemChoice}{ i\in\{1,2\}\\ \pi\langle\epsilon, s_1 \ndchoice  s_2\rangle = s_i}{\langle \epsilon, s_1 \ndchoice s_2\rangle \xrightarrow{1}_\pi  \langle \epsilon',s_i\rangle}
\qquad
\ncondrule{If1}{\epsilon(e)=\textit{true}}{\langle\epsilon, \code{if}\: e\: \{ s_1\}\ \code{else}\ \{ s_2\} \rangle\\  \xrightarrow{1}_\pi \langle\epsilon, s_1\rangle}
\qquad
\ncondrule{If2}{\epsilon(e)=\textit{false}}{\langle\epsilon,
        \code{if}\: e\: \{ s_1\}\ \code{else}\ \{ s_2\} \rangle
      \\\xrightarrow{1}_\pi \langle \epsilon, s_2\rangle}
\\\\
\ncondrule{ProbChoice1}{\epsilon(e)=p \quad 0 \leq p \leq 1}{\langle \epsilon, s_1\pchoice{e} s_2\rangle \xrightarrow{p}_\pi  \langle \epsilon,s_1\rangle}
\qquad
\ncondrule{While1}{\epsilon(e)=\textit{true}}{\langle\epsilon, \code{while}\: e\: \{
        s\}\rangle \xrightarrow{1}_\pi \langle\epsilon, s;\code{while}\: e\: \{
      s\}\rangle}
\\\\
\ncondrule{ProbChoice2}{\epsilon(e)=p \quad 0 \leq p \leq 1}{\langle \epsilon, s_1\pchoice{e} s_2\rangle \xrightarrow{1-p}_\pi  \langle \epsilon,s_2\rangle}
\qquad
\ncondrule{While2}{\epsilon(e)=\textit{false}}{\langle\epsilon, \code{while}\: e\: \{ s\}\rangle \xrightarrow{1}_\pi \langle \epsilon, \code{skip}\rangle}\\[-6pt]
\end{array}$$
\caption{\label{fig:semantics}An MDP-semantics for pGCL.}
\end{figure}

\section{PDL: Probabilistic Dynamic Logic}\label{sec:pdl}
The probabilistic dynamic logic \PDL was introduced by Pardo
\etal~\cite{2022-pdl-ictac} as a specification language for
probabilistic programs in \pGCL. \PDL builds on dynamic logic
\cite{harel00dynlog}, a modal logic in which logical formulae with
programs in the modalities can be used to express reachability
properties. Our formulation of \PDL here differs from our previous
work~\cite{2022-pdl-ictac} by incorporating \emph{symbolic
  updates}~\cite{beckert16dynamic}, a technique for representing state
change in forward symbolic execution that is well-known from the KeY
verification system~\cite{key}.

\subsection{Syntax of \PDL}
Given sets \(X\) of program variables and \(L\) of logical variables
disjoint from \(X\), let \LIT\ denote the well-formed atomic formulae
built using constants, program and logical variables.  For every
\(l\!\in\!L\), let \dom{l}\ denote the domain of \(l\).  Let $x$ range
over $X$ and $t$ over well-formed terms, which for our purposes are
\pGCL expressions where also logical variables are allowed.

The formulae $\phi$ of probabilistic dynamic logic (\PDL) are defined
inductively as the smallest set generated by the following grammar.
\begin{align*}\label{eq:pdl-syntax}
  \phi \quad ::=& \quad \LIT
  \gmid \neg \phi
  \gmid \phi_1\land \phi_2
  \gmid \forall l \cdot \phi
  \gmid \dlbox{s}_\pfun\:\phi
  \gmid \{U\}\phi
  \\
  U \quad ::=& \quad \mathsf{empty} \gmid x~\mapsto~t
\end{align*}
where \(\phi\) ranges over \PDL formulae, \(l\! \in\! L\) over logical
variables, $s$ is a \pGCL program with variables in \(X\), and
$\pfun \colon \State \to [0,1]$ is an expectation assigning values in
$[0,1]$ to initial states of the program \(s\).  The logical operators
\(\to\), \(\lor\) and \(\exists\) are derived in terms of \(\neg\),
\(\wedge\) and \(\forall\) as usual.   As usual, state formulae are \PDL
formulae without the box-modality; we denote state formulae by \FOL.

A formula can be constructed by applying symbolic updates $U$ to
formulae $\phi$; i.e., $\{U\}\phi$ is a well-formed formula. Note that
formulae that include symbolic updates are typically used in
intermediate steps in a proof system based on forward symbolic
execution; symbolic updates are typically not part of user-provided
specifications. Symbolic updates are syntactic representations of term
substitutions for program variables, which keep track of symbolic
state changes within a proof branch. The empty update $\mathsf{empty}$
denotes no change, the update $x~\mapsto~t$ denotes a state change (or
substitution~\cite{key}) where the program variable $x$ has the value
of $t$. The update application $\{U\}\phi$ applies the update $U$ to
$\phi$.

\subsection{Semantics of \PDL}

We extend valuations to also map logical variables \(l\in L\) to
values in \dom{l} and let \(\epsilon\models_\LIT \phi\) denote
standard satisfaction, expressing that \(\phi\in\LIT\) holds in
valuation \(\epsilon\). From now on, we equate valuations and states,
writing, for instance, \(\epsilon \in \State\). Though states consist
not only of valuations, but also of program locations, the locations
are not relevant for interpreting \PDL formulae.

First, we define the semantics of updates as a function from program
state to program state. Let $\epsilon(t)$ be the evaluation of a term
$t$ in a state $\epsilon$.  The updated states
$\updatesemantics{U}(\epsilon)$ for a given substitution $U$ and state
$\epsilon$ are given by
\begin{align*}
\updatesemantics{\mathsf{empty}}(\epsilon) &= \epsilon\\
\updatesemantics{x\mapsto t}(\epsilon) & =\epsilon[x\mapsto \epsilon(t)]
\end{align*}
We define satisfiability of \emph{well-formed formulae} in \PDL\ as follows:
\looseness = -1

\begin{definition}[Satisfiability of \PDL Formulae]
  \label{def:satisfaction}
  Let \(\phi\) be a well-formed \PDL formula, $\pi$ range over
  policies, \(l \! \in \! L\), \(\pfun : \State \to [0,1]\) be an
  expectation lower bound, and $\epsilon$ be a valuation defined for
  all variables mentioned in $\phi$.  The \emph{satisfiability} of a
  formula $\phi$ in a model $\epsilon$, denoted
  $\epsilon\models \phi$, is defined inductively as follows:
  \begin{align*}
    & \epsilon \models \phi
    & \textnormal{iff}\quad
    & \phi\in \LIT \quad\textnormal{and}\quad  \epsilon \models_\LIT \phi \\
    & \epsilon \models \phi_1\land \phi_2
    & \textnormal{iff}\quad
    & \epsilon \models \phi_1 \quad\textnormal{and}\quad \epsilon \models \phi_2
    \\
    & \epsilon \models \neg \phi
    & \textnormal{iff}\quad
    & \textnormal{not } \epsilon \models \phi
    \\
    & \epsilon \models \forall l \cdot  \phi
    & \textnormal{iff}\quad
    & \epsilon \models \phi \subst{l}{v} \textnormal{ for each } v\in \dom{l}
    \\
    & \epsilon \models \{U\}\phi
    & \textnormal{iff}\quad
    & \updatesemantics{U}(\epsilon) \models \phi
    \\
    & \epsilon \models \dlbox{s}_\pfun \phi
    & \textnormal{iff}\quad
    & \pfunepsilon \leq \Expectation{\embed \phi}
    \textnormal{ where the expectation is taken in $\MDP_s$}
  \end{align*}
\end{definition}

\noindent
For $\phi\in \LIT$, \(\models_\LIT\) can be used to check satisfaction
just against the valuation of program variables since \(\phi\) is
well-formed.  In the case of universal quantification, the
substitution replaces logical variables with constants.  The last case
(p-box) is implicitly recursive, since the characteristic function
\embed\phi\ refers to the satisfaction of \(\phi\) in the final states
of \(s\).  We use the characteristic function \embed\phi\ as the
reward function on the final state of $\MDP_s$.  In other words, the
satisfaction of a p-box formula $\pbox{s}{\pfun}\:\phi$ captures a
lower bound on the probability of \(\phi\) holding after the program
\(s\). Consequently, \PDL supports specification and reasoning about
probabilistic reachability properties in almost surely terminating
programs.  We use $\models \pbox{s}{\pfun} \, \phi$ to denote that a
formula is \emph{valid}, \textit{i.e.},
$\epsilon \models \pbox{s}{\pfun}\, \phi $ for all valuations
$\epsilon$.

\begin{proposition}[Properties of \PDL \cite{2022-pdl-ictac}]\label{prop:properties}
  \textnormal{
  Let $s,s_1,s_2$ be \pGCL programs, $\phi$ a \PDL formula, and $\epsilon$ a valuation.
  \begin{enumerate}[label=(\roman*)]
    \item\label{item:termination} (termination) $\epsilon \models \pbox{\code{skip}}{\func{1}} \phi$ if and only if $\epsilon \models \phi$;
    \item\label{item:inaction} (inaction) $\epsilon \models \pbox{s}\pfun \phi$ if and only if $\epsilon \models \pbox{\code{skip};\ s}\pfun \phi$;
    \item\label{item:assign} (assign) $\epsilon \models \pbox{x := e;\ s}\pfun \phi$ if and only if $\epsilon[x \mapsto \epsilon(e)] \models \pbox s \pfun\phi$;
    \item\label{item:unilowb} (universal lower bound) $\epsilon \models \pbox{s}{\func{0}}\phi$;
    \item\label{item:quantweak} (quantitative weakening) if $\epsilon\models \pbox s\pfun\phi$ and $\pfun' \leq \pfun$ then $\epsilon \models \pbox s{\pfun'}\phi$;
    \item\label{item:demchoice} (demonic choice) $\epsilon \models \pbox{s_1}{\pfun}\phi$ and $\epsilon \models \pbox{s_2}{\pfun}\phi$ if and only if $\epsilon \models \pbox{s_1\sqcap s_2}\pfun \phi$;
    \item\label{item:pchoice} (probabilistic choice) if $\epsilon \models \pbox{s_1}{\pfun_1}\phi$ and $\epsilon \models \pbox{s_2}{\pfun_2}\phi$ then $\epsilon \models \pbox{s_1 \pchoice e s_2}\pfun \phi$, where $\pfun = \epsilon(e)\pfun_1 + (1-\epsilon(e))\pfun_2$;
    \item\label{item:iftrue} (if true) if $\epsilon \models e \land \pbox{s_1}\pfun \phi$ then $\epsilon \models \pbox{\code {if}\ (e)\ \{s_1\}\ \code{else}\ \{s_2\}}\pfun \phi$;
    \item\label{item:iffalse} (if false) if $\epsilon \models \neg e \land \pbox{s_2}\pfun \phi$ then $\epsilon \models \pbox{\code {if}\ (e)\ \{s_1\}\ \code{else}\ \{s_2\}}\pfun \phi$; and
    \item\label{item:loopunfold} (loop unfold) $\epsilon \models \pbox{\code {if}\ (e)\ \{s;\ \code{while}\ (e)\ \{s\}\}\ \code{else}\ \{\code{skip}\}}\pfun \phi$ if and only if $\epsilon \models \pbox{\code {while}\ (e)\ \{s\}}\pfun \phi$.
  \end{enumerate}
}
\end{proposition}
Note that in the above, arithmetic operations are lifted point-wise when
applied to expectation lower bounds (functions \pfun, \pfun').

\section{A Proof System for \PDL}\label{sec:rules}

\paragraph{Judgments.}
Let $\Gamma$ be a set of formulae and $\phi$ a singular formula.  We
write a \emph{judgment} $\judgment\Gamma\phi$ to state that the
formula $\bigwedge \Gamma \rightarrow \phi$ is valid, i.e., that
$\epsilon \models \bigwedge \Gamma \rightarrow \phi$ for every
$\epsilon$.  When defining rules, we denote with $\mathcal{U}$ a
nested sequence of update applications of any depth, for example the
schematic formula $\mathcal{U}\phi$ matches
$\{U_1\}\{U_2\}\dots\{U_n\}\phi$ for some natural $n$ and updates
$U_1,\dots,U_n$.

\paragraph{Probabilistic constraints.}
Our proof system is going to work with constraints on probabilities.
Let $\Phi$ be a set of \PDL-formulae, $\mathcal{U}$ a symbolic update
and $\mathit{eq}$ an (in)equality over probability variables.  A
\PDL-constraint
$\pdleq{\Phi}{\mathcal{U}}{\mathit{eq}}^\Phi_\mathcal{U}$ expresses
that $\mathcal{U}(\mathit{eq})$ must hold in any state represented by
the symbolic update $\mathcal{U}$ such that $\mathcal{U}(\Phi)$.  In
the sequel, we introduce the simplifying assumption that that
probabilistic expressions in \PDL do not depend on state variables, in
which case we can simplify constraints to simple (in)equalities, e.g.,
$\pfun \leq \func{1}$.

\paragraph{A probabilistic dynamic logic calculus for pGCL.}
We now introduce the inference rules for dynamic logic formulae of the
form $\Gamma \Rightarrow \mathcal{U}\pbox{ s}{\pfun}\:\phi$,
expressing that if $\Gamma$ holds in the initial state of some program
execution, the probability of reaching a state in which $\phi$ holds
from a state described by the symbolic update $\mathcal{U}$ by
executing $s$, is at least $\pfun$ (as before applied to the initial
state). Thus, if the symbolic update $\mathcal{U}$ is empty, we are in
the initial state and $s$ is the entire program to be analyzed.  The
application of the inference rules creates a proof tree in which
\PDL-constraints are generated as side conditions to rule
applications.  In the end, the proof-tree can be closed if its
\PDL-constraints are satisfiable.

\Cref{fig:new_rules} shows the inference rules for \pGCL---we omit
inference rules for FOL connectives as they are standard,
see~\cite{beckert16dynamic} for details.  The rules in
\Cref{fig:new_rules} are the syntax-driven, aiming to eliminate the
weakening rule.  Instead weakening is implicitly applied in the rules
that add a \PDL-constraint. Let us first consider the rule for
\code{skip}, to explain the rule format. The rules symbolically
execute the first statement of the box-modality, generating a premise
that must hold for the symbolic execution of the remaining program
$s$.  Note that \refrule{skip} does not add any probabilistic
constraint. In the rules \refrule{empty1} and \refrule{empty0}, the
empty program is denoted by a \code{skip}-statement without a
continuation $s$. In these rules, our aim is to check that the
postcondition $\phi$ holds in the current state, represented by the
symbolic update $\mathcal{U}$ with a given probability $\pfun$. If
$\pfun \leq \func{0}$, we leave the proof tree open in rule
\refrule{empty0}. If $\pfun \geq \func{1}$, we close the tree if the
formula holds in rule \refrule{empty1}.

Branching is expressed through probabilistic and demonic choice. Rule
\refrule{demonChoice} captures the demonic choice by taking the worst
case of the two branches, where $\pfun_1$ and $\pfun_2$ capture the
probability of the postcondition $\phi$ holding for each of the
branches. The \PDL-constraint
$\pdleq{\Gamma}{\mathcal{U}}{\pfun \leq \mathbf{min}(\pfun_1,
  \pfun_2)}$ captures this worst-case assumption. Rule
\refrule{probChoice} similarly captures the probabilistic choice, here
the \PDL-constrain combines the probability f selecting a branch with
the probability of reaching the postcondition in that
branch. Conditional branching over an \code{if}-statement is captured
by the rule \refrule{if}, which is standard and simply adds the
condition for selecting each branch in the state captured by the
symbolic update $\mathcal{U}$ to the precondition of each premise.

Loops are captured by rule \refrule{loopUnroll} which simply
unfolds a \code{while}-statement into a conditional statement.

Asserting $\judgment\Gamma\phi$ expresses existence of a proof tree
with root $\judgment{\Gamma}{\phi}$ whose set of side conditions is
satisfiable.
\begin{theorem}[Soundness]\label{thm:soundness}
  For all $\Gamma$ and $\phi$, if $\judgment\Gamma\phi$ then $\models (\bigwedge \Gamma) \rightarrow \phi$.
\end{theorem}
\begin{proof}
  By induction on the height of the proof tree that justifies the
  judgment $\judgment\Gamma\phi$, with an analysis of the rule that
  justifies the root of the proof tree.  All rules except \refrule{if}
  are of the form where the premise is $\judgment \Gamma {\phi_1}$ and
  the conclusion is $\judgment \Gamma {\phi_2}$.  In such cases, it
  suffices to prove that if $\phi_1$ holds then $\phi_2$ holds.
  Throughout the proof, we use
  $\mathcal U = \{U_1\}\{U_2\}\dots\{U_n\}$ for some $n$ and symbolic
  updates $U_1,\dots,U_n$.  Then, using \Cref{def:satisfaction},
  $\epsilon \models \mathcal U \phi$ if and only if
  $(\updatesemantics{U_n}\circ\updatesemantics{U_{n-1}}\circ\dots\circ\updatesemantics{U_1})(\epsilon)
  \models \phi$.  We will write $\rho_{\mathcal U}$ for the state
  transformation
  $\updatesemantics{U_n}\circ\updatesemantics{U_{n-1}}\circ\dots\circ\updatesemantics{U_1}$
  \begin{itemize}
    \item For rule \refrule{empty1} $\epsilon \models \mathcal U (\phi)$ iff $\rho_{\mathcal U} (\epsilon) \models \phi$ iff $\rho_{\mathcal U} (\epsilon) \models \pbox {\code{skip}}{\func{1}} \phi$ iff $\epsilon \models \mathcal U \pbox{\code{skip}}{\func{1}} \phi$.
      Here, we use \cref{prop:properties}\ref{item:termination}.
    \item For rule \refrule{skip}, $\epsilon \models \mathcal U \pbox{s}{\pfun} \phi$ iff $\rho_{\mathcal U}(\epsilon) \models \pbox{s}{\pfun} \phi$ iff $\rho_{\mathcal U} (\epsilon) \models \pbox{\code{skip};s}{\pfun} \phi$ iff $\epsilon \models \mathcal U \pbox{\code{skip};s}{\pfun} \phi$.
      Here, we use \cref{prop:properties}\ref{item:inaction}.
    \item For rule \refrule{assign}, $\epsilon \models \mathcal U \{x\mapsto e\}\pbox s\pfun \phi$ iff $\rho_{\mathcal U}(\epsilon) \models \{x \mapsto e\} \pbox s\pfun \phi$.
      Now write $\epsilon' := \rho_{\mathcal U}(\epsilon)$.
      Then $\epsilon' \models \{x \mapsto e\} \pbox s\pfun \phi$ iff $\epsilon'[x \mapsto \epsilon'(e)] \models \pbox s\pfun \phi$ iff $\epsilon' \models \pbox{x := e; s}p \phi$, using \cref{prop:properties}\ref{item:assign}.
    \item For rule \refrule{empty0}, $p=0$ is a universal lower bound for any formula: $\rho_{\mathcal U}(\epsilon) \models \pbox{\code{skip}}{\func{0}} \phi$, by \cref{prop:properties} (universal lower bound).
    \item For rule \refrule{demonChoice}, let $\epsilon'$ be arbitrary.
      By assumption, $\epsilon' \models \updates\pbox{s_1;s}{\pfun_1} \phi$ and $\epsilon' \models \updates \pbox{s_2;s}{\pfun_2} \phi$.
      That is, with $\epsilon := \rho_{\updates}(\epsilon')$, $\epsilon \models \pbox{s_1;s}{\pfun_1}\phi$ and $\epsilon \models \pbox{s_2;s}{\pfun_2}$.
      Put $\pfun := \min(\pfun_1,\pfun_2)$.
      By \cref{prop:properties}\ref{item:demchoice}, $\epsilon \models \pbox{s_1;s \sqcap s_2;s}\pfun \phi$, meaning $\pfun \leq \Expectation{\embed \phi}$, taken in $\MDP_{s'}$ where $s'=s_1;s\sqcap s_2;s$.
      Analyzing the rules in \Cref{fig:semantics} and using the definition of expectation, this means that both $\pfun \leq \Expectation{\embed \phi}$ taken in $\MDP_{s_1;s}$ and $\pfun \leq \Expectation{\embed \phi}$ taken in $\MDP_{s_2;s}$.
      Now assume for contradiction that $\pfun > \Expectation{\embed \phi}$ taken in $\MDP_{(s_1\sqcap s_2);s}$.
      Then there is a policy $\pi$ such that $\pfun > \ExpectedV{\embed \phi}$ in $\MDP_{(s_1\sqcap s_2);s}$.
      Using Rules~\textsc{DemChoice} and~\textsc{Composition2}, then, depending on the policy but without loss of generality, $\pfun > \ExpectedV{\embed \phi}$ also in $\MDP_{s_1;s}$.
      But this contradicts what we claimed before, that $\pfun \leq \Expectation{\embed \phi}$ in $\MDP_{s_1;s}$.
      Conclude that $\epsilon \models \pbox{(s_1\sqcap s_2);s}\pfun \phi$ and so $\epsilon' \models \updates \pbox{(s_1\sqcap s_2);s}\pfun \phi$.
      By \cref{prop:properties}\ref{item:quantweak}, this is true for any $\pfun \leq \min(\pfun_1,\pfun_2)$, so we are done.
    \item For rule \refrule{probChoice}, let $\epsilon'$ be arbitrary.
      By assumption, $\epsilon' \models \updates\pbox{s_1;s}{\pfun_1}\phi$ and $\epsilon' \models \updates\pbox{s_2;s}{\pfun_2}\phi$, meaning $\epsilon \models \pbox{s_1;s}{\pfun_1}\phi$ and $\epsilon \models \pbox{s_2;s}{\pfun_2}\phi$ with $\epsilon = \rho_\updates(\epsilon')$.
      With $\pfun := e \cdot \pfun_1 + (1-e) \cdot \pfun_2$, by \cref{prop:properties}\ref{item:pchoice}, $\epsilon \models \pbox{s_1;s \pchoice{e} s_2;s}\pfun \phi$.
      This means that $\pfun \leq \Expectation {\embed\phi}$ in the MDP of $s_1;s\pchoice es_2;s$, and analyzing Rules~\textsc{ProbChoice1}, \textsc{ProbChoice2}, and~\textsc{Composition2}, we know that
      $$ \underbrace{\Expectation{\embed \phi}}_{\textup{in }s_1;s \pchoice e s_2;s} = \epsilon(e)\cdot \underbrace{\Expectation{\embed \phi}}_{\textup{in }s_1;s} + (1-\epsilon(e))\cdot \underbrace{\Expectation{\embed \phi}}_{\textup{in }s_2;s} = \underbrace{\Expectation{\embed \phi}}_{\textup{in }(s_1\pchoice es_2);s} $$
      we conclude that also $\epsilon \models \pbox{(s_1\pchoice es_2);s}\pfun\phi$, and hence, $\epsilon' \models \updates \pbox{(s_1\pchoice e s_2);s}\pfun\phi$.
      The argument is generalizable to any $\pfun \leq e \cdot \pfun_1 + (1-e) \cdot \pfun_2$ (\cref{prop:properties}\ref{item:quantweak}), so we are done.
    \item For rule \refrule{if}, let $\epsilon'$ be arbitrary and assume the premises of the rule hold.
      If $\epsilon' \not \models \Gamma$ there is nothing to prove.
      Otherwise, w.l.o.g., $\epsilon' \models \updates(e)$ and therefore $\epsilon' \models \updates \pbox{s_1;s}\pfun \phi$.
      Hence, writing $\epsilon = \rho_\updates(\epsilon)$, we have $\epsilon \models e \land \pbox{s_1;s}\pfun \phi$.
      By \cref{prop:properties}\ref{item:iftrue}, then, $\epsilon \models \pbox{\code{if}\ (e)\{s_1;s\}\ \code{else}\ \{s_2;s\}}\pfun \phi$.
      This means that $\pfun \leq \Expectation{\embed \phi}$ in the MDP of $\code{if}\ (e)\{s_1;s\}\ \code{else}\ \{s_2;s\}$.
      Using Rules~\textsc{If1} and~\textsc{Composition2} and the fact that $\epsilon \models e$, also $\pfun \leq \Expectation{\embed \phi}$ for $s_1;s$.
      Then, using Rule~\textsc{If1} and $\epsilon\models e$ again, $\pfun \leq \Expectation {\embed \phi}$ in $(\code{if}\ (e) \{s_1\}\ \code{else}\ \{s_2\}); s$, so that $\epsilon \models \pbox{(\code{if}\ (e) \{s_1\}\ \code{else}\ \{s_2\}); s}\pfun \phi$, and so $\epsilon' \models \updates\pbox{(\code{if}\ (e) \{s_1\}\ \code{else}\ \{s_2\}); s}\pfun \phi$.
    \item Finally, for rule \refrule{loopUnroll}, let $\epsilon'$ be arbitrary and let $\epsilon = \rho_\updates(\epsilon')$.
      There are two cases to consider:
      \begin{itemize}
        \item If $\epsilon \models e$ (meaning $\epsilon(e)=\textup{true}$) then: $\epsilon \models \pbox{\code {while}\ (e)\ \{s_b\};\ s}\pfun \phi$ iff $\epsilon \models \pbox{s_b;\ \code{while}\ (e)\ \{s_b\};\ s}\pfun \phi$ (using Rules~\textsc{While1} and~\textsc{Composition2} and definition of expectation) iff $\epsilon \models \pbox{\code{if}\ (e)\ \{s_b;\ \code{while}\ (e)\ \{s_b\};\ s\}\ \code{else}\ s}\pfun \phi$ (using Rules~\textsc{If1} and~\textsc{Composition2}).
        \item Otherwise, if $\epsilon \not \models e$, similar reasoning with Rules~\textsc{While2},~\textsc{If2}, and~\textsc{Composition1} show that $\epsilon \models \pbox{\code {while}\ (e)\ \{s_b\};\ s}\pfun \phi$ if and only if $\epsilon \models \pbox{\code{if}\ (e)\ \{s_b;\ \code{while}\ (e)\ \{s_b\};\ s\}\ \code{else}\ s}\pfun \phi$
      \end{itemize}
      In either case, we have shown that $\epsilon' \models \updates \pbox{\code {while}\ (e)\ \{s_b\};\ s}\pfun \phi$ if and only if $\epsilon'\models \updates\pbox{\code{if}\ (e)\ \{s_b;\ \code{while}\ (e)\ \{s_b\};\ s\}\ \code{else}\ s}\pfun \phi$
  \end{itemize}
\end{proof}

\begin{figure}[t]
  \centering
  $\begin{array}{l@{\qquad}l}
     \INFER{skip}{
  \judgment{\Gamma} {\mathcal{U}\pbox{s}{\pfun}~\phi}}{
  \judgment{\Gamma} {\mathcal{U}\pbox{\code{skip};~s}{\pfun}~\phi}}
&

  \INFER[\pdleq{\Gamma}{\mathcal{U}}{\pfun \leq \mathbf{min}(\pfun_1, \pfun_2)}]{demonChoice}{
  \judgment{\Gamma} {\mathcal{U}\pbox{s_1;s~}{\pfun_1}~\phi}\\
     \judgment{\Gamma} { \mathcal{U}\pbox{s_2;~s}{\pfun_2}~\phi}
     }{
     \judgment{\Gamma} {\mathcal{U}\pbox{(s_1 \sqcap s_2);~s}{\pfun}~\phi}
     }
     \\\\
     \INFER{assign}{
     \judgment{\Gamma} {\mathcal{U}\{x\mapsto e\}\pbox{s}{\pfun}~\phi}
     }{
     \judgment{\Gamma} {\mathcal{U}\pbox{x := e ;~s}{\pfun}~\phi}
     }

     &

       \INFER[\pdleq{\Gamma}{\mathcal{U}}{\pfun \leq e\cdot \pfun_1 + (1-e)\cdot \pfun_2}]{probChoice}{
  \judgment{\Gamma} { \mathcal{U}\pbox{s_1;~s}{\pfun_1}~\phi}\\
  \judgment{\Gamma} { \mathcal{U}\pbox{s_2;~s}{\pfun_2}~\phi}
}{
  \judgment{\Gamma} { \mathcal{U}\pbox{s_1 \pchoice{e} s_2;~s}{\pfun}~\phi}
}
     \\\\

   \INFER[\pdleq{\Gamma}{\mathcal{U}}{\pfun \doteq 1}]{empty1}{
  \judgment{\Gamma} {\mathcal{U}(\phi)}
}{
  \judgment{\Gamma} {\mathcal{U}\pbox{\code{skip}}{\pfun}~\phi}
     }

     &

       \INFER{if}{
  \judgment{\Gamma, \mathcal{U}(e)} {\mathcal{U}\pbox{s_1;~s}{\pfun}~\phi}\\
  \judgment{\Gamma, \neg \mathcal{U}(e)} {\mathcal{U}\pbox{s_2;~s}{\pfun}~\phi}
}{
  \judgment{\Gamma} {\mathcal{U}\pbox{\code{if}\ (e)\{s_1\}\ \code{else}\ \{s_2\};~s}{\pfun}~\phi}
}
     \\\\

   \INFER[\pdleq{\Gamma}{\mathcal{U}}{\pfun \doteq 0}]{empty0}{

}{
  \judgment{\Gamma} { \mathcal{U}\pbox{\code{skip}}{\pfun}~\phi}
     }

     &
\INFER{loopUnroll}{
\judgment{\Gamma} {\mathcal{U}\pbox{\code{if}\ (e)\ \{s_b;~\code{while}\ (e)\ \{s_b\};~s\}~\code{else}\ s\}}{\pfun}~\phi}
}{
  \judgment{\Gamma} {\mathcal{U}\pbox{\code{while}\ (e)\ \{s_b\};~s}{\pfun}~\phi}
     }

     \end{array}
$
\caption{Symbolic execution rules.\label{fig:new_rules}}     
\end{figure}

\paragraph{Example.}
  We illustrate the use of the inference rules in~\cref{fig:new_rules}
  by considering one of the branches of the proof tree for
  property~\cref{eq:monty} of the \emph{Monty Hall}
  game~(\cref{sec:example}).  In~\cref{sec:implementation}, we use our
  prototype implementation to automatically generate the complete
  proof.  The property we are interested in proving is as follows (the
  code for \code{Monty_Hall} is in~\cref{code:monty-hall}):
  $$
  \judgment
  {\mathit{switch} = \mathit{true}}
  {[\text{\lstinline|Monty_Hall|}]_\pfun (\mathit{choice} = \mathit{prize})}
  $$
  First, we apply the rule for non-deterministic choice. 
  For convenience, we use $\phi \triangleq (\mathit{choice} = \mathit{prize})$.
  $$
  \rulemacro
  {
    \judgment
    {\mathit{switch} = \mathit{true}}
    {[\text{\lstinline|prize:=0; ...|}]_{\pfun_0} \phi}
    \\
    \judgment
    {\mathit{switch} = \mathit{true}}
    {[\text{\lstinline|prize:=1; ...|}]_{\pfun_1} \phi}
    \\
    \judgment
    {\mathit{switch} = \mathit{true}}
    {[\text{\lstinline|prize:=2; ...|}]_{\pfun_2} \phi}
  }
  {
    \judgment
    {\mathit{switch} = \mathit{true}}
    {[\text{\lstinline|Monty_Hall|}]_\pfun \phi}
  }
  {
    \pfun \leq \min(\pfun_0,\pfun_1,\pfun_2)
  }
  $$
  The rule adds 3 different premises, one for each path of the non-deterministic choice.
  In what follows, we focus on the proof branch for $\pfun_0$.
  The next program statement is an assignment, thus we apply the rule {\sc Assign}.
  $$
  \rulemacro
  {
    \judgment
    {\mathit{switch} = \mathit{true}}
    {\{\mathit{prize} \mapsto 0\}
      [\text{\lstinline|choice:=0|} \pchoice{1/3} \text{\lstinline|(choice:=1|} \pchoice{1/2} \text{\lstinline|choice:=2|}\text{\lstinline|);...|}]_{\pfun_0} 
      \phi}
  }
  {
    \judgment
    {\mathit{switch} = \mathit{true}}
    {[\text{\lstinline|prize:=0; ...|}]_{\pfun_0} \phi}
  }
  $$ 
  This rule simply added the update $\{\mathit{prize} \mapsto 0\}$.
  In the following, we use $\Gamma$ and $\updates$ to denote, at a point in the proof tree, the set of formulae in the precedent of a judgment and the set of updates, respectively.
  Next we apply the rule to resolve the probabilistic choice:
  \\
  $$
  \rulemacro
  {
    \judgment
    {\Gamma}
    {\ \updates[\text{\lstinline|choice:=0; ...|}]_{\pfun_{00}}\phi}    
    \\
    \judgment
    {\Gamma}
    {\ \updates[\text{\lstinline|(choice:=1|} \pchoice{1/2} \text{\lstinline|choice:=2|}\text{\lstinline|);...|}]_{\pfun_{01}} \phi}
  }
  {
    \judgment
    {\Gamma}
    {\ \updates[\text{\lstinline|choice:=0|} \pchoice{1/3} \text{\lstinline|(choice:=1|} \pchoice{1/2} \text{\lstinline|choice:=2|}\text{\lstinline|);...|}]_{\pfun_0} \phi}
  }
  {
    \pfun_0 \leq 1/3\pfun_{00} + 2/3\pfun_{01}
  }
  $$
  We continue the example with the proof tree for the $\pfun_{00}$ premise.
  We apply {\sc Assign} so that we add the assignment as a symbolic update $\updates = \{\mathit{prize} \mapsto 0, \mathit{choice} \mapsto 0\}$.
  Then, we apply the rule for if-statements:
  $$
  \rulemacro
  {
    \judgment
    {\Gamma, \overbrace{\updates(\code{prize = choice})}^{0=0}}
    {[s_0 \code{;...}]_{\pfun_{00}}\phi}
    \quad
    \judgment
    {\Gamma, \overbrace{\neg \, \updates(\code{prize = choice})}^{0\not=0}}
    {[s_1 \code{;...}]_{\pfun_{00}}\phi}
  }
  {
    \judgment
    {\Gamma}
    {[\text{\lstinline[mathescape=true]|if (prize = choice) \{$s_0$\} else \{$s_1$\};...|}]_{\pfun_{00}}\phi}
  }
  {}
  $$
  The branch for the right premise above can be closed as we have derived false ($0\not=0$).
  For the left premise, we apply the rule for resolving the non-deterministic choice $[$\lstinline[mathescape=true]|open := (prize+1)
  This step generates a new constraint $\pfun_{00} \leq \min(\pfun_{000}, \pfun_{001})$.
  We continue discussing the left branch (\lstinline[mathescape=true]|open := (prize+1)
  In this case, we include a new update $\{\mathit{open} \mapsto 1\}$, as $\updates((\mathit{prize}+1) \% 3) = (0+1) \% 3 = 1$.
  Since we have $\mathit{switch} = \mathit{true}$, for the final if-statement (\lstinline[mathescape=true]|if (switch) choice := (2*choice - open)
  All in all, we get $\judgment{\Gamma}{\{ \mathit{prize} \mapsto 0, \mathit{open} \mapsto 1, \mathit{choice} \mapsto 2 \} (\mathit{prize} = \mathit{choice})}$.
  Since, given this update, the property does not hold, we close the branch by applying {\sc Empty0}, i.e., $\pfun_{000} = 0$.

  In summary, this proof branch has produced the following set of constraints: $\pfun \leq \min(\pfun_0,\pfun_1,\pfun_2)$, $\pfun_0 \leq 1/3\pfun_{00} + 2/3\pfun_{01}$, $\pfun_{00} \leq \min(\pfun_{000}, \pfun_{001})$ and $\pfun_{000} = 0$.
  Similar reasoning can be applied in the remaining branches, which will produce additional sets of constraints. 
  In~\cref{sec:implementation}, we use Crowbar to automatically generate the complete proof and constraints.
  Furthermore, Crowbar outsources the set of constraints to an SMT solver to determine whether there are values for the different $\pfun_i$ that satisfy all constraints.

\paragraph{Discussion.}

The above proof system exploits the results of our earlier work
\cite{2022-pdl-ictac} to build an automated deduction system.  In this
system, we have placed a rather simple instrument for reasoning about
loops. We recall that a classic, non-probabilistic loop-invariant rule
for partial correctness has the following structure (cast in a dynamic
logic flavor \cite{key}):

\begin{equation*}
  \INFER{loopInvariant}{
    \judgment{\Gamma} {\mathcal U~I} \quad
    \judgment{\Gamma, \mathcal{U'} (\neg e \land I)}{\mathcal U' \pbox{s}{}~\phi} \quad
    \judgment{\Gamma, \mathcal{U'} (e \land I)}{\mathcal{U'} \pbox{s_b}{}~I)}
  }{
    \judgment{\Gamma} {\mathcal U\pbox{\code{while}\ (e)\ \{s_b\};~s}{}~\phi}
  }
\end{equation*}
where $\mathcal{U'}$ denotes some symbolic update such that the
invariant holds (thus abstracting from the number of iterations of
the loop).
Using our syntax (and exploiting the probability lower bound of one to
encode a qualitative box formula), this rule corresponds to:
\begin{equation*}
  \INFER{loopQual}{
    \judgment{\Gamma} {\mathcal U~I} \quad
    \judgment{\Gamma, \mathcal U' (\neg e \land I)}{ \mathcal U' \pbox{s}{\func{1}}~\phi} \quad
    \judgment{\Gamma,\mathcal U' (e \land I)}{\mathcal U' \pbox{s_b}{\func{1}}~I)}
  }{
    \judgment{\Gamma} {\mathcal U\pbox{\code{while}\ (e)\ \{s_b\};~s}{\func{1}}~\phi}
  }
\end{equation*}
The key aspect is to invent both an invariant and the substitution at
termination, which (roughly) correspond to proposing a loop invariant
in Hoare logics. This rule generalizes in a straightforward way to the
case when the loop body itself is non-probabilistic. Then, it suffices
to show that the expectation of the loop successor \( s \) is
preserved by the loop, obtaining:
\begin{equation*}
  \INFER{loopPreserving}{
    \judgment{\Gamma} {\mathcal U~I} \quad
    \judgment{\Gamma, \mathcal U' (\neg e \land I)}{\mathcal U' \pbox{s}{\pfun}~\phi} \quad
    \judgment{\Gamma, \mathcal U' (e \land I)} {\mathcal U' \pbox{s_b}{\func{1}}~I)}
  }{
    \judgment{\Gamma} {\mathcal U\pbox{\code{while}\ (e)\ \{s_b\};~s}{\pfun}~\phi}
  }
\end{equation*}
In general, weakest-precondition style rules can be encoded, and as
our next step, we intend to experiment with existing rules from that
space in actual proofs\,\cite{kaminski19phd}.

\section{Implementation}\label{sec:implementation}
We have implemented the above proof procedure\footnote{The sources for
  our prototype are available at
  \url{https://github.com/Edkamb/crowbar-tool/tree/PDL}.} in
Crowbar~\cite{DBLP:journals/scp/KamburjanSR23}, a modular symbolic
execution engine which uses the SMT-solver Z3~\cite{demoura2008} under
the hood to solve arithmetic constraints: Crowbar performs the
symbolic execution and uses Z3 to discharge formulas without
modalities. Crowbar was originally developed to experiment with
deductive proof systems for the active object modeling language
ABS~\cite{johnsen10fmco}. For rapid prototyping, we have therefore
used the Crowbar front-end for ABS and encoded \pGCL into the main
block of ABS programs; we use the general annotation
format~\cite{DBLP:journals/scp/SchlatteJKT22} for ABS programs to
capture probabilistic aspects of \PDL as annotations of ABS
programs. The probabilistic constraints accumulated during the
symbolic execution of the \pGCL program are then passed to Z3 if the
proof can otherwise be closed.

The main block itself is annotated with two \emph{specification
  elements}: first the annotation \abs{[Spec: Ensures(e)]} captures
the post-condition and the annotation \abs{[Spec: Prob(e)]} captures
its probability.  Furthermore, we use annotations of statements to
encode probabilistic and demonic choice-operators of \pGCL using an
ABS branching statement.  Demonic choice is expressed with a
\mbox{\abs{[Spec: Demonic]}} annotation on the branching statement,
whose guard is then ignored. Probabilistic choice is similarly
expressed with a \abs{[Spec: Prob(e)]} annotation.

\begin{figure}[t]
\begin{abscode}
[Spec: Ensures(choice == prize)]
[Spec: Prob(2/3)]
{
Int prize = -1; Int choice = -1; Int open = -1;
Int sw = 1;

[Spec: Demonic]
if( True ) { prize = 0; }
else { [Spec: Demonic] if( True ) { prize = 1; }
                      else { prize = 2; } }

[Spec: Prob(1/3)]
if( True ) { choice = 0; }
else { [Spec: Prob(1/2)] if( True ) { choice = 1; }
                        else { choice = 2; } }

if(prize==choice){
    [Spec: Demonic] if( True ) { open = (prize + 1 ) % 3; }
                    else { open = (prize + 2) % 3; }}
else {open = ((2 * prize) - choice) % 3;}

if(sw==1){choice = ((2 * choice) - open) % 3;} else { skip; }
}
\end{abscode}
    \caption{The Monty Hall game in the ABS encoding for Crowbar.}
    \label{fig:monty}
\end{figure}

To illustrate the encoding of \pGCL and \PDL into the ABS
representation of Crowbar, \Cref{fig:monty} shows the encoding of the
Monty Hall game introduced in \Cref{sec:example}. If
\texttt{sw} is set to 1, denoting that the player always switches, the
proof obligation can be discharged, as expected. If \texttt{sw} is set
to 0, denoting that the player never switches, the proof attempt
fails.

\section{Related and Future Work}
The symbolic execution proof system for probabilistic dynamic logic
presented in this paper formulated in the style of KeY \cite{key}, and
its symbolic execution proof system for sequential Java. We
specifically use their technique for symbolic updates
\cite{beckert16dynamic} in the formulation of our symbolic execution
proof system. In contrast to their work, our deductive verification
system addresses probabilistic programs by accumulating probabilistic
constraints, which are then resolved using Z3.

Voogd \etal \cite{voogd23qest} developed a symbolic execution
framework for probabilistic programs, building on Kozen's work on the
semantics of probabilistic programs with random variables
\cite{kozen79}. However, this work, inspired by de Boer
and Bonsangue's work on formalized symbolic execution
\cite{deboer2021symbolic}, does not extend into a deductive proof
system as discussed in this paper. It further focuses on probabilistic
programs with random variables rather than probabilistic and demonic
choice as introduced in \pGCL. In future work, we aim to integrate
ideas from this paper into our reasoning framework; in particular, it
would be interesting to add support for observe-statements. We refer
to Voogd \etal \cite{voogd23qest} for a detailed discussion of how
different semantics for probabilistic programs relate to symbolic
execution.

Katoen \etal \cite{schroer23oopsla} developed a generic deductive
proof system for probabilistic programs. Their work uses an encoding
into a dedicated verification engine for probabilistic reasoning,
called Caesar. In line with systems like Viper~\cite{muller16vmcai},
they develop a intermediate representation language. To address
probabilistic reasoning, their intermediate language is probabilistic
and generates verification conditions for an SMT solver. Our work uses
Crowbar, which is also a modular verification system.  However, Caesar
is tailored for weakest-precondition-style backwards reasoning, while
our work targets forward reasoning by means of symbolic execution.

Our investigation shows that invariant based reasoning for loops is
surprisingly subtle for probabilistic programs. Here, Joost-Pieter
Katoen and his collaborators have again led the way (e.g.,
\cite{mcIver18popl}). We hope some of his work on probabilistic loops
can carry over to backwards reasoning in a weakest pre-expectation
framework to forward reasoning in a symbolic execution setting, but it
is an open question today---an open question that we aim to
understand---exactly how forward and backwards reasoning about
probabilistic loops relate.

\section{Conclusion}
This paper reports on work in progress towards a proof system for
\PDL, a probabilistic dynamic logic for specifying properties about
probabilistic programs. A nice feature of \PDL is that it is closed
under logical operators such as first-order connectives and
quantifiers, and that it has a model-theoretic semantics in terms of
a satisfiability-relation. Our approach to a proof system for \PDL
uses forward reasoning by combining symbolic execution with a
constraint solver. We have outlined a proof system for deductive
verification based on symbolic execution which collects constraints
about probabilities as side conditions. These are then forwarded to
the constraint solver.  For our prototype implementation, we have used
Crowbar, a modular deductive verification engine based on symbolic
execution, and the SMT-solver Z3 to solve constraints on
probabilities.

We would like to end this paper by thanking Joost-Pieter Katoen not
only for all his outstanding work on deductive verification of
probabilistic programs, but also for interesting discussions on the
relationship between backwards and forwards reasoning in this
setting. A next step for us will be to understand if (and hopefully,
how) some of the work conducted by Katoen and colleagues on reasoning
about loops can adapted to a forward-reasoning framework.

\medskip

\noindent
\textbf{Acknowledgments.}
This work was partly funded by the EU project \textsf{SM4RTENANCE}
(grant no.~101123423). We are grateful to Asmae Heydari Tabar for
discussions on symbolic execution and Crowbar.


\begin{thebibliography}{10}
\providecommand{\url}[1]{\texttt{#1}}
\providecommand{\urlprefix}{URL }
\providecommand{\doi}[1]{https://doi.org/#1}

\bibitem{key}
Ahrendt, W., Beckert, B., Bubel, R., H{\"{a}}hnle, R., Schmitt, P.H., Ulbrich,
  M. (eds.): Deductive Software Verification - The KeY Book - From Theory to
  Practice, Lecture Notes in Computer Science, vol. 10001. Springer (2016),
  \url{https://doi.org/10.1007/978-3-319-49812-6}

\bibitem{baier03tse}
Baier, C., Haverkort, B.R., Hermanns, H., Katoen, J.-P.: Model-checking algorithms
  for continuous-time {M}arkov chains. {IEEE} Trans. Software Eng.
  \textbf{29}(6),  524--541 (2003),
  \url{https://doi.org/10.1109/TSE.2003.1205180}

\bibitem{baier08book}
Baier, C., Katoen, J.-P.: Principles of model checking. {MIT} Press (2008)

\bibitem{baier99concur}
Baier, C., Katoen, J.-P., Hermanns, H.: Approximate symbolic model checking of
  continuous-time {M}arkov chains. In: Baeten, J.C.M., Mauw, S. (eds.) Proc.\
  10th International Conference on Concurrency Theory ({CONCUR}'99). Lecture
  Notes in Computer Science, vol.~1664, pp. 146--161. Springer (1999),
  \url{https://doi.org/10.1007/3-540-48320-9\_12}

\bibitem{batz24popl}
Batz, K., Biskup, T.J., Katoen, J.-P., Winkler, T.: Programmatic strategy
  synthesis: Resolving nondeterminism in probabilistic programs. Proc. {ACM}
  Program. Lang.  \textbf{8}({POPL}),  2792--2820 (2024),
  \url{https://doi.org/10.1145/3632935}

\bibitem{batz23tacas}
Batz, K., Chen, M., Junges, S., Kaminski, B.L., Katoen, J.-P., Matheja, C.:
  Probabilistic program verification via inductive synthesis of inductive
  invariants. In: Sankaranarayanan, S., Sharygina, N. (eds.) Proc.\ 29th
  International Conference on Tools and Algorithms for the Construction and
  Analysis of Systems ({TACAS} 2023). Lecture Notes in Computer Science, vol.
  13994, pp. 410--429. Springer (2023),
  \url{https://doi.org/10.1007/978-3-031-30820-8\_25}

\bibitem{batz21popl}
Batz, K., Kaminski, B.L., Katoen, J.-P., Matheja, C.: Relatively complete
  verification of probabilistic programs: an expressive language for
  expectation-based reasoning. Proc. {ACM} Program. Lang.  \textbf{5}({POPL}),
  1--30 (2021), \url{https://doi.org/10.1145/3434320}

\bibitem{beckert16dynamic}
Beckert, B., Klebanov, V., Wei{\ss}, B.: Dynamic logic for {J}ava. In: Ahrendt
  et~al.  \cite{key}, pp. 49--106,
  \url{https://doi.org/10.1007/978-3-319-49812-6\_3}

\bibitem{deboer2021symbolic}
de~Boer, F.S., Bonsangue, M.: Symbolic execution formally explained. Formal
  Aspects of Computing  \textbf{33}(4),  617--636 (2021),
  \url{https://doi.org/10.1007/s00165-020-00527-y}

\bibitem{dehnert017cav}
Dehnert, C., Junges, S., Katoen, J.-P., Volk, M.: A {S}torm is coming: {A} modern
  probabilistic model checker. In: Majumdar, R., Kuncak, V. (eds.) Proc.\ 29th
  International Conference on Computer Aided Verification ({CAV} 2017). Lecture
  Notes in Computer Science, vol. 10427, pp. 592--600. Springer (2017),
  \url{https://doi.org/10.1007/978-3-319-63390-9\_31}

\bibitem{feng23oopsla}
Feng, S., Chen, M., Su, H., Kaminski, B.L., Katoen, J.-P., Zhan, N.: Lower bounds
  for possibly divergent probabilistic programs. Proc. {ACM} Program. Lang.
  \textbf{7}({OOPSLA1}),  696--726 (2023),
  \url{https://doi.org/10.1145/3586051}

\bibitem{Haehnle22a}
H{\"{a}}hnle, R.: Dijkstra's legacy on program verification. In: Apt, K.R.,
  Hoare, T. (eds.) Edsger Wybe Dijkstra: His Life, Work, and Legacy, pp.
  105--140. {ACM} / Morgan {\&} Claypool (2022),
  \url{https://doi.org/10.1145/3544585.3544593}

\bibitem{harel00dynlog}
Harel, D., Kozen, D., Tiuryn, J.: Dynamic Logic. Foundations of Computing, MIT
  Press (Oct 2000)

\bibitem{hensel22sttt}
Hensel, C., Junges, S., Katoen, J.-P., Quatmann, T., Volk, M.: The probabilistic
  model checker {S}torm. Int. J. Softw. Tools Technol. Transf.  \textbf{24}(4),
   589--610 (2022), \url{https://doi.org/10.1007/s10009-021-00633-z}

\bibitem{johnsen10fmco}
Johnsen, E.B., H{\"{a}}hnle, R., Sch{\"{a}}fer, J., Schlatte, R., Steffen, M.:
  {ABS:} {A} core language for abstract behavioral specification. In:
  Aichernig, B.K., de~Boer, F.S., Bonsangue, M.M. (eds.) Proc.\ 9th
  International Symposium on Formal Methods for Components and Objects ({FMCO}
  2010). Lecture Notes in Computer Science, vol.~6957, pp. 142--164. Springer
  (2010), \url{https://doi.org/10.1007/978-3-642-25271-6\_8}

\bibitem{junges24fmsd}
Junges, S., {\'{A}}brah{\'{a}}m, E., Hensel, C., Jansen, N., Katoen, J.-P.,
  Quatmann, T., Volk, M.: Parameter synthesis for {M}arkov models: covering the
  parameter space. Formal Methods Syst. Des.  \textbf{62}(1),  181--259 (2024),
  \url{https://doi.org/10.1007/s10703-023-00442-x}

\bibitem{DBLP:journals/scp/KamburjanSR23}
Kamburjan, E., Scaletta, M., Rollshausen, N.: Deductive verification of active
  objects with {Crowbar}. Sci. Comput. Program.  \textbf{226},  102928 (2023),
  \url{https://doi.org/10.1016/j.scico.2023.102928}

\bibitem{kaminski19phd}
Kaminski, B.L.: Advanced weakest precondition calculi for probabilistic
  programs. Ph.D. thesis, {RWTH} Aachen University, Germany (2019),
  \url{http://publications.rwth-aachen.de/record/755408}

\bibitem{kaminski18jacm}
Kaminski, B.L., Katoen, J.-P., Matheja, C., Olmedo, F.: Weakest precondition
  reasoning for expected runtimes of randomized algorithms. J. {ACM}
  \textbf{65}(5),  30:1--30:68 (2018), \url{https://doi.org/10.1145/3208102}

\bibitem{kozen79}
Kozen, D.: Semantics of probabilistic programs. In: Proc.\ 20th Annual
  Symposium on Foundations of Computer Science. pp. 101--114. {IEEE} Computer
  Society (1979), \url{https://doi.org/10.1109/SFCS.1979.38}

\bibitem{mciver05book}
McIver, A., Morgan, C.: Abstraction, Refinement and Proof for Probabilistic
  Systems. Monographs in Computer Science, Springer (2005),
  \url{https://doi.org/10.1007/b138392}

\bibitem{mcIver18popl}
McIver, A., Morgan, C., Kaminski, B.L., Katoen, J.-P.: A new proof rule for
  almost-sure termination. Proc. {ACM} Program. Lang.  \textbf{2}({POPL}),
  33:1--33:28 (2018), \url{https://doi.org/10.1145/3158121}

\bibitem{demoura2008}
de~Moura, L., Bj{\o}rner, N.: Z3: An efficient {SMT} solver. In: Ramakrishnan,
  C.R., Rehof, J. (eds.) Proc.\ 14th International Conference on Tools and
  Algorithms for the Construction and Analysis of Systems ({TACAS} 2008).
  Lecture Notes in Computer Science, vol.~4963, pp. 337--340. Springer (2008),
  \url{https://doi.org/10.1007/978-3-540-78800-3\_24}

\bibitem{muller16vmcai}
M{\"{u}}ller, P., Schwerhoff, M., Summers, A.J.: Viper: {A} verification
  infrastructure for permission-based reasoning. In: Jobstmann, B., Leino,
  K.R.M. (eds.) Proc.\ 17th International Conference on Verification, Model
  Checking, and Abstract Interpretation ({VMCAI} 2016). Lecture Notes in
  Computer Science, vol.~9583, pp. 41--62. Springer (2016),
  \url{https://doi.org/10.1007/978-3-662-49122-5\_2}

\bibitem{2022-pdl-ictac}
Pardo, R., Johnsen, E.B., Schaefer, I., Wąsowski, A.: A specification logic
  for programs in the probabilistic guarded command language. In: Proc.\ 19th
  International Colloquium on Theoretical Aspects of Computing ({ICTAC}'22).
  Lecture Notes in Computer Science, vol. 13572, pp. 369--387. Springer (2022),
  \url{https://doi.org/10.1007/978-3-031-17715-6\_24}

\bibitem{puterman}
Puterman, M.L.: Markov Decision Processes. Wiley (2005)

\bibitem{DBLP:journals/scp/SchlatteJKT22}
Schlatte, R., Johnsen, E.B., Kamburjan, E., Tapia~Tarifa, S.L.: The {ABS}
  simulator toolchain. Sci. Comput. Program.  \textbf{223},  102861 (2022),
  \url{https://doi.org/10.1016/j.scico.2022.102861}

\bibitem{schroer23oopsla}
Schr{\"{o}}er, P., Batz, K., Kaminski, B.L., Katoen, J.-P., Matheja, C.: A
  deductive verification infrastructure for probabilistic programs. Proc. {ACM}
  Program. Lang.  \textbf{7}({OOPSLA2}),  2052--2082 (2023),
  \url{https://doi.org/10.1145/3622870}

\bibitem{voogd23qest}
Voogd, E., Johnsen, E.B., Silva, A., Susag, Z.J., Wąsowski, A.: Symbolic
  semantics for probabilistic programs. In: Jansen, N., Tribastone, M. (eds.)
  Proc.\ 20th International Conference on Quantitative Evaluation of Systems
  ({QEST} 2023). Lecture Notes in Computer Science, vol. 14287, pp. 329--345.
  Springer (2023), \url{https://doi.org/10.1007/978-3-031-43835-6\_23}

\end{thebibliography}
\end{document}